\documentclass[a4paper,11pt,reqno]{article}

\usepackage[utf8]{inputenc}%(only for the pdftex engine)
\usepackage[left =2.7cm, right = 2.7 cm]{geometry}

\usepackage{lipsum}
\usepackage{amsfonts}
\usepackage{graphicx}
\usepackage{algorithmic}
\usepackage{subcaption}
\usepackage{tabularx}
\usepackage{amsmath,amssymb,amsthm}

\usepackage{xcolor}

\usepackage{tabularx}
\usepackage{booktabs}
\usepackage{makecell}
\usepackage{longtable}
\usepackage[ruled,vlined]{algorithm2e}
\usepackage{algorithmic}
\usepackage{listings}
\usepackage{etoolbox}% http://ctan.org/pkg/etoolbox
\usepackage{lipsum}% http://ctan.org/pkg/lipsum

\usepackage{url}

\usepackage{amssymb}
\usepackage{amsfonts}
\usepackage{amsmath}
\usepackage{comment}
\usepackage{tabularx}
\usepackage{booktabs}
\usepackage{makecell}
\usepackage{longtable}
\usepackage{graphicx}
\usepackage{caption}
\usepackage{lipsum}
\usepackage{xcolor}
\usepackage{parskip}
\usepackage{caption}
\usepackage[hidelinks]{hyperref}

\usepackage{authblk}

\usepackage{enumitem}
\setlist[enumerate]{leftmargin=.5in}
\setlist[itemize]{leftmargin=.5in}

\newcommand{\N}{\mathbb{N}}
\newcommand{\Z}{\mathbb{Z}}

\newtheorem{theorem}{Theorem}

\newtheorem{lemma}{Lemma}
\newtheorem{proposition}{Proposition}

\theoremstyle{definition}

\newtheorem{definition}{Definition}
\theoremstyle{remark}
\newtheorem{remark}{Remark}

\author{Simran Tinani}
\author{Joachim Rosenthal}
\affil{Institute of Mathematics, University of Zurich, Switzerland}

\title{A Deterministic Algorithm for \\ the Discrete Logarithm Problem
in a Semigroup}

\begin{document}
\maketitle
  \begin{abstract} The discrete logarithm problem in a finite group is the basis for many protocols in cryptography. The best general algorithms which solve this problem have time complexity of $\mathcal{O}(\sqrt{N}\log N)$, and a space complexity of $\mathcal{O}(\sqrt{N})$ where $N$ is the order of the group. (If $N$ is unknown, a simple modification would achieve a time complexity of $\mathcal{O}(\sqrt{N}(\log N)^2)$.) These algorithms require the inversion of some group elements or rely on finding collisions and the existence of inverses, and thus do not adapt to work in the general semigroup setting. For semigroups, probabilistic algorithms with similar time complexity have been proposed. The main result of this paper is a deterministic algorithm for solving the discrete logarithm problem in a semigroup. Specifically, let $x$ be an element in a semigroup having finite order $N_x$. The paper provides an algorithm, which, given any element $y\in \langle x \rangle $, provides all natural numbers $m$ with $x^m=y$, and has time complexity $O(\sqrt{N_x}(\log N_x)^2)$ steps. The paper also gives an analysis of the success rates of the existing probabilistic algorithms, which were so far only conjectured or stated loosely. \end{abstract}

\section{Introduction}
\label{intro}
Let $G$ be a group and assume $x,y\in G$ are two elements of the group. We refer to $x$ as the base element. The discrete
logarithm problem (referred to henceforth as DLP) asks for the computation 
of an integer $m\in \Z$ (assuming such integers exist) such that $x^m=y$.
The DLP plays an important role in a multitude of algebraic and number theoretic cryptographic systems. Its use was introduced in the Diffie-Hellman protocol for public key exchange \cite{di76} and has since seen a tremendous amount of development, generalisations and extensions \cite{me97}. Many modern-day systems for public key exchange use the discrete logarithm problem in a suitable group. The most commonly used groups have been the multiplicative group of finite fields and the group of points on an elliptic curve. The DLP in Jacobians of hyperelliptic curves 
and more general abelian varieties has also been studied extensively~\cite{co06}.

In this paper, we compute complexities using group multiplications as one fundamental step. Thus, an exponentiation $x^e$ is performed in $\mathcal{O}(\log e)$ steps. We will use the fact that for two lists of length $n$ in which a match exists, a match can be found in $\mathcal{O}(n\log n)$ steps using standard sorting and searching algorithms (for details, the interested reader may refer to \cite{cormen2009introduction}). For a general finite group of order $N$, there exist algorithms that solve the DLP in $\mathcal{O}(\sqrt{N}\log N)$ steps. Such algorithms are said to produce a square root attack. The most well-known examples are Shank's Baby Step-Giant Step algorithm \cite{shanks1971class} and the Pollard-Rho algorithm \cite{pollard1978monte}. 
Note that Shank's algorithm is a deterministic algorithm having time complexity
$\mathcal{O}(\sqrt{N}\log N)$ space complexity $\mathcal{O}(\sqrt{N})$.
In contrast, Pollard's algorithm is a probabilistic algorithm having time complexity
$\mathcal{O}(\sqrt{N}\log N)$ group multiplications and space complexity $\mathcal{O}(1)$. If $N$ is unknown, a simple modification of these algorithms would achieve a time complexity of $\mathcal{O}(\sqrt{N}(\log N)^2)$.

Elliptic curve groups have been widely implemented in practice since for a carefully selected elliptic curve group the best known classical algorithm for solving DLP has running time $\mathcal{O}(\sqrt{N}\log N)$, 
where $N$ is the group order. This is in contrast to many other finite groups such as the multiplicative
group of a finite field and the group of invertible matrices over a finite field where algorithms with  subexponential running time are known \cite{ad93}.

In cryptography the Diffie-Hellman protocol using a finite group has been generalized 
to situations where the underlying problem is a discrete logarithm problem 
in a semigroup or even to situations where a semigroup acts on a set~\cite{ka13,ma07}. 
The interested reader will find more material in a recent survey by Goel et al.~\cite{go20}.

It is naturally interesting to ask whether the DLP also has a square root attack in more generalized structures such as semigroups. Here, we define a semigroup as any set of elements with an associative binary operation. Since the best algorithms for the DLP all make use of the existence of inverses, it is unclear whether they can be generalized to a semigroup. However, when a special type of semigroup element, called a torsion element, is used as the base, it turns out that the DLP is reducible in polynomial time to the DLP in a finite group. A torsion element is one whose powers eventually repeat to form a cycle, and will be defined more precisely in Section \ref{basics}. This section also elaborates more on why the standard collision-based algorithms are not directly adaptable to the semigroup case. A semigroup in which every element is torsion is called a torsion semigroup.

The DLP in semigroups with a torsion base element, in a classical setting, was first discussed by Chris Monico \cite{monico2002} in 2002, and later in a paper by Banin and Tsaban \cite{banin} in 2016. While the discussion in the present paper is entirely on classical algorithms, it is also worth mentioning the paper \cite{quantum}, where the authors independently provide a quantum algorithm that solves the DLP in a torsion semigroup.

Both the algorithm of Monico and the one of Banin and Tsaban are probabilistic and
might fail with low probability. Further, some of their methods are heuristic, dependent on an oracle or some additional assumption, and their success rates and expected number of steps are either conjectured or stated loosely. It is therefore of interest to come up with 
an algorithm which deterministically computes the discrete logarithm in a semigroup.
In this regard we like to make some analogy to the problem of determining if an integer 
is a prime number, a problem of great importance in cryptography. Nowadays in practice 
the algorithm of Miller and Rabin~\cite{mi76,ra80} has been implemented for many years.
Still it was a great result when Agrawal, Kayal and Saxena~\cite{ag04} came up
with a deterministic polynomial time algorithm to achieve this goal. 

A key step in finding the discrete logarithm in a semigroup is computing the cycle length of an element. Both the algorithms of \cite{banin} and \cite{monico2002} rely on computing some multiple of the cycle length, and then removing ``extra" factors by taking gcd's until the cycle length is obtained. Once the cycle length value is obtained, the discrete logarithm may easily be computed with a few more simple steps. While Monico does not provide further elaboration on how this is done, the paper by Banin and Tsaban bridges this knowledge gap by showing how the problem is reduced to a DLP in a group once the cycle length and start values are known. Denote by $N_x$ the order of $x$ (formally defined in Definition~\ref{ord_elt}). The complexity of the algorithm in \cite{banin} is $\mathcal{O}(\sqrt{N_x}(\log N_x)^2\log\log N_x)$, and the of the one in \cite{monico2002} is $\mathcal{O}(\sqrt{N_x}(\log N_x)^2)$. While both of the existing methods seem to succeed with high probability for practical values, we show that the process of taking successive gcd's/factors is unnecessary, and that one can deterministically find the cycle length. The main contribution of this paper will be a deterministic algorithm for computing the discrete logarithm of an element $y$ in some semigroup $S$ with respect to some torsion base element $x\in S$. The time complexity of our algorithm is $\mathcal{O}(\sqrt{N_x}(\log N_x)^2)$.

The paper is structured as follows: After providing preliminaries and basic 
definitions in Section \ref{basics}, we will analyse 
in Section \ref{probalgo} the success rates and expected number of steps involved in the probabilistic algorithms for cycle length by Banin and Tsaban (Algorithm \ref{banin_algo}) and Monico (Algorithm \ref{monico_algo}).

In Section \ref{det}, which is the main section of this paper, we provide a deterministic algorithm to calculate the cycle length $L_x$ of a torsion element $x$ of a semigroup and thus to also solve the DLP, without the use of an oracle. This algorithm has complexity $\mathcal{O}\left( \sqrt{N_x}\cdot (\log N_x)^2\ \right).$ For completeness, we will also demonstrate the use of Pohlig--Hellman algorithm \cite{pohlig1978improved} for a semigroup. 

\section{Preliminaries}\label{basics}

A semigroup $S$ is a set together with an associative binary operation.
Like in group theory where a torsion group consists of elements of finite order only we define:

 \begin{definition}[Torsion Element] Let $S$ be a semigroup. An element $x \in S$ is called a torsion element if the sub-semigroup $\langle x \rangle:=\{x^k\mid k\in \N \}$ generated by $x$, is finite. $S$ is called a torsion semigroup if every $x\in S$ is a torsion element.
 
 \end{definition}

 Throughout the paper the following definitions will be assumed:
 
  \begin{definition}[Cycle Start] Let $x \in S$. The cycle start $s_x$ of $x$ is defined as the smallest positive integer such that $x^{s_x} = x^{b} $ for some $b\in \N$, $b>s_x$.
\end{definition}

 \begin{definition}[Cycle Length] Let $x \in S$. The cycle length $L_x$ of $x$ is defined as the smallest positive integer such that $x^{s_x + L_x} = x^{s_x}$.
\end{definition}

 \begin{definition}[Element order]\label{ord_elt} Let $x \in S$. With notation as above, we define the order 
 $N_x$ of $x$ as the cardinality of the sub-semigroup $\langle x \rangle$. Note that 
 $N_x=s_x+L_x-1.$
\end{definition}

 \begin{definition}[Semigroup DLP] Let $S$ be a semigroup and $x \in S$. The semigroup DLP is defined as follows. Given $y\in \langle x \rangle:=\{x^k\mid k\in \N \}$, find all $m\in N$ such that $x^m=y$.
 \end{definition} 
 
 We state below a key result first proved in \cite{banin}.

\begin{lemma}[\cite{banin}]\label{group_within} Let $S$ be a semigroup and $x \in S$ be an element with cycle start $s_x$. The set of powers \ $G_x=\{ x^{s_x+k} \mid k\geq 0\}$ \ of $x$ forms a finite cyclic group. The identity element of $G_x$ is given by $x^{tL_x}$, where $t$ is the minimum positive integer such that $x^{tL_x} \in G_x$.
  \end{lemma} 

The following result is stated in \cite{monico2002} in a slightly different formulation. We provide an equivalent proof based on the group structure of $G_x$.
 
 \begin{lemma}[\cite{monico2002}]
 Let $x \in S$ have cycle start $s_x$ and cycle length $L_x$. For all integers $n, m \geq s_x$, we have $x^m = x^n \iff n\equiv m \mod L_x$. 
  \end{lemma} 
\begin{proof}
We can assume without loss of generality that $n \geq m$, and so we can write $n = m + kL_x+ u $, with $k \geq 0$ and $0\leq u <L_x$.
First suppose that $n\equiv m \mod L_x$, i.e. $u = 0$. Since $m, n \geq s_x$, we have $x^n = x^{m+kL_x} = x^m$. 

Conversely, if $x^n = x^m$,
write $n_1 = n-s_x \geq 0$, and $m_1 = m-s_x \geq 0$. We have \[x^{s_x + m_1}= x^{s_x+n_1} = x^{s_x+m_1 + kL_x + u} = x^{s_x+m_1+u}. \]
Now, without loss of generality, $m_1 \geq s_x$, because if not, one can always increment $m_1$ and $n_1$ by multiples of $L_x$ until this happens. So, we can assume that $x^{m_1}$ lies in $G_x$ and is thus invertible. We multiply by the inverse on both sides to finally get \[x^{s_x} = x^{s_x+u}.\] 

Thus, we must have $u=0$ or $n \equiv m \mod L_x$, as required. 

\end{proof}

\begin{remark}
It becomes clear from the above discussion that the standard collision-based algorithms for order and discrete log computations in a group do not adapt directly to a general semigroup. Collision-based algorithms for the computation of the order $N$ of a group element $x$ (for instance, see \cite{Sutherland2007OrderCI}) are based on the principle that whenever $N$ can be expressed in the form $N = A-B$ for non-negative integers $A$ and $B$, the collision $x^A = x^B$ always occurs. However, this principle does not work in a semigroup, where there are two independent components of the order. More specifically, for a semigroup element $x$ with cycle length $L_x$ and cycle start $s_x$, whenever $L_x$ may be expressed in the form $A-B$ for non-negative integers $A$ and $B$, the equality $x^A = x^B$ holds if and only if $A, B\geq s_x$. As an example, consider a semigroup element $x$ with cycle length $L_x = 12$ and cycle start $s_x=5$. Then, $L_x= 15-3$, but $x^{15} \neq x^3 $. Thus without prior knowledge of the cycle start, the semigroup order $N_x$ or cycle length $L_x$ cannot directly be found using the same collision-based algorithms for groups.
 
 Similarly, collision-based algorithms fail for discrete log computations in a semigroup. As an example, consider a semigroup element $x$ with cycle length $L_x = 15$ and cycle start $s_x = 10$, and suppose that the discrete log of $y=x^5$ is to be found. Then $y\cdot x^6 = x^{11} = x^{26}$ is obtained as a collision. However, unlike in the group case, the conclusion $y = x^{26-6}= x^{20}$ is wrong since $x^5 \neq x^{20}$. This happens because even though $x$ is torsion and forms a cycle of powers, it is not invertible. 
\end{remark}

This concludes the prerequisite knowledge on torsion elements in semigroups. In the next section, we study the existing probabalistic algorithms for cycle lengths, and analyse their assumptions, working and complexities.

\section{Existing Probabalistic Agorithms}\label{probalgo}

\subsection{Banin and Tsaban's Algorithm}\label{banin}

In this section, we study the probabalistic algorithm described in \cite{banin} for computing the cycle length of a torsion element in a semigroup. While the authors of the original paper describe their theory only for torsion semigroups, it will become clear that the same discussion holds true for any semigroup when the base element chosen is torsion.

Let $S$ be a semigroup and $x$ be a torsion element of $S$. Let $s_x$ denote the cycle start of $x$ and $L_x$ its cycle length. Then, recall from Lemma \ref{group_within} that $G_x:=\{x^{s_x}, x^{s_x+1}, \ldots, x^{s_x+L_x-1}\}$ is a cyclic group, and that it has order $L_x$. The authors of \cite{banin} assume the availability of a `Discrete Logarithm Oracle' for the group $G_x$, which returns values $\log_x h$ for $h \in G_x$. They state that these values need not be smaller than the group order but are polynomial in the size of $G_x$ and the element $x$. The representation of the identity in $G_x$ is unknown, and a method to compute inverses is not available.

The authors claim that the well-known algorithms for discrete logarithm computations in groups do not explicitly require inverses, or can easily be modified to work without the use of inverses. While it is true that these algorithms make use of mainly the existence of inverses rather than their explicit computation, we believe that the fact that easy modification is possible is not immediate without some justification. In fact, it will become clear in the later sections that the modified Baby-Step-Giant-Step algorithm devised by Monico \cite{monico2002} (and also the deterministic algorithm presented in Section \ref{det}) is a crucial and non-trivial part of any such modification.

We make the following observation from the proof of Lemma 1 found in \cite{banin}. For any $k\geq 0$, denote by $v_k$ the smallest positive integer such that 
 $$v_k L_x \geq 2s_x+k.$$
 We then have $x^{v_k L_x-s_x-k} \in G_x$ and \begin{equation}\label{inverseformula} x^{s_x + k}x^{v_k L_x-s_x-k} = x^{v_k L_x} = x^{tL_x},\end{equation} so the inverse of the element $x^{s_x + k}$ of $G_x$ is given by $x^{v_k(L_x)-s_x-k}$. In particular, the computation of inverses requires prior knowledge of the cycle start. As will be explained below, the cycle start may be computed only once the value of the cycle length is known, using a binary search. This explains why the authors insist that their Discrete Logarithm Oracle does not need to use the computation of inverses. 
 
Below, we describe Algorithm \ref{banin_algo}, which is the algorithm suggested in \cite{banin} to compute the order of the group $G_x$, i.e. the cycle length $L_x$ of $x$.

\begin{algorithm}[ht]  
\hspace*{\algorithmicindent} \textbf{Input} {A semigroup $S$ and a torsion element $x \in S$; a DLP oracle for groups} \\
 \hspace*{\algorithmicindent} \textbf{Output} {The cycle length $L_x$ of $x$}
\begin{algorithmic}[1]
\STATE Initialize $i,j, g, L_x \leftarrow 1$, $N>>s_x+L_x$. Fix bounds $r>1, s>1$. \\
\STATE \textbf{while $j<s$}
\begin{enumerate}
\item Fix a random $z \in \{\lfloor M/2 \rfloor, \ldots, M\}$ and set $h=x^z$.
\STATE \textbf{while $i<r$}
\begin{enumerate}
\item Choose a random number $k_i>0$. 
\item Use the DLP oracle to compute $k_i^{\prime}=\log_{h}(h^{k_i})$. 
\item Set $g \leftarrow \gcd\limits_{j\leq i} (k_j-k_j^{\prime}) = \gcd\left(\gcd\limits_{j<i} (k_j-k_{j}^{\prime}), \ k_i -{k_i}^\prime \right)$. 
\item Set $i \leftarrow i+1$. \end{enumerate}
\STATE \textbf{end while}
\STATE Set $L_x \leftarrow lcm(L_x, g)$, $j \leftarrow j+1$.
 \end{enumerate}
  \STATE Return $L_x$.
 \caption{Banin-Tsaban Algorithm for Cycle Length } \label{banin_algo}
 \end{algorithmic}
\end{algorithm}

We first note that the authors state complexities in terms of $L_x$, which are valid when the bound $N$ for $N_x$ is known. If the algorithm fails for a value of $N$, the authors suggest to double $N$ and try again. In this case, which we will assume from now on, we assert that the complexities need to be taken in terms of $N_x$ instead of $L_x$. The oracle may be assumed to have the standard complexity of {$O(\sqrt{N_x}\log N_x)$} steps for discrete logarithm calculations. 
Step 2.2(c) takes $\mathcal{O}(\log (\max_{j \leq i}(k_j-k_i))=\mathcal{O}(\log N_x)$ integer operations by the assumption on the oracle, which does not contribute to the total complexity. Thus, the total complexity of step (2.2) comes from the oracle alone, and is $\mathcal{O}( \sqrt{N_x}\log N_x)$. Now, the authors of \cite{banin} remark that $r$ and $s$ can be taken to satisfy $r=\mathcal{O}(1)$ and $s=\mathcal{O}(\log \log N_x)$. Thus, the total complexity is $\mathcal{O}(\log N_x)$ times the complexity of Algorithm \ref{banin_algo}, and thus $\mathcal{O}(\log \log(N_x) \log N_x)$ times the complexity of step (2.2). 
Therefore, we get the total complexity of $ \mathcal{O}(\log \log N_x (\log N_x)^2 \sqrt{N_x})$.

Finally, in Algorithm \ref{cycle_start_algo}, we present the application of the binary search method to find the cycle start once $L_x$ is known. This algorithm is formulated as below for this purpose in \cite{banin}, though the idea to use a binary search is also originally mentioned in \cite{monico2002}.

 \begin{algorithm}[ht] \caption{Calculating Cycle Start (Binary Search)}\label{cycle_start_algo}
\hspace*{\algorithmicindent} \textbf{Input} {A semigroup element $x$ with cycle length $L_x$} \\
 \hspace*{\algorithmicindent} \textbf{Output} {Cycle start $s_x$ of $x$}
\begin{algorithmic}[1]
 \STATE Initialize $s_x \leftarrow 1$ \\
 \STATE \textbf{while $x^{s_x+L_x}\neq x^{s_x}$ do} \\
     \quad \quad $ s_x \leftarrow 2s_x$ \\
\STATE \textbf{end while} 
\STATE Set $a \leftarrow s_x/2$
\STATE \textbf{while $|a-s_x| \geq 2$} \\
 $c \leftarrow (a+s_x)/2$ \\
\quad \textbf{if $x^{c+L_x} \neq x^c$ then} \\
\quad \quad $a \leftarrow c$ \\
\quad \textbf{else} \\
\quad \quad $ s_x \leftarrow c$
 \STATE \textbf{end while}
\end{algorithmic}
\end{algorithm}

\begin{lemma}\label{cycle_start}
Let $N_x$ be the order of the element $x$. Then Algorithm \ref{cycle_start_algo} requires 
$$ \mathcal{O}\left((\log N_x)^2\right). $$
steps.
\end{lemma}
\begin{proof}
Each of Steps (2) and (5) involves $\mathcal{O}(\log N_x)$ rounds, each of which computes requires $\mathcal{O}(\log N_x)$ semigroup multiplications and one comparison. The total complexity is thus $ \mathcal{O}\left((\log N_x)^2\right)$. 
\end{proof}

\subsection{Monico's Algorithm}\label{monico}

 In his PhD thesis \cite{monico2002}, Chris Monico provides a probabilistic algorithm (described below as Algorithm \ref{monico_algo}) that calculates the cycle length of an element in a finite ring of order $N$. This algorithm makes use of the multiplicative semigroup structure of the finite ring, and of the availability of the explicit bound $N$ for every cycle length, and is in fact applicable to any  semigroup where such a bound $N$ is available. In this subsection, we analyse this algorithm, provide a more concrete bound on its success rate, and compute its complexity in terms of $N$. We will discuss this algorithm in terms of torsion semigroups, as opposed to finite rings.

\begin{algorithm}[ht] \caption{Monico's Baby-Step Giant-Step for Cycle Length}\label{monico_algo}
\hspace*{\algorithmicindent} \textbf{Input} {A finite semigroup $S$ with $\lvert S \rvert = N$ and an element $x \in S$} \\
 \hspace*{\algorithmicindent} \textbf{Output} {The cycle length $L_x$ of $x$}
\begin{algorithmic}[1]
\STATE Set $m = \lceil\sqrt{N}\rceil$. Choose a prime $q > N$. \\

\STATE For $0 \leq i \leq m$, compute and store in a table the pairs $(i; x^{q+im})$. \\
Sort the table by the second component. \\

\STATE Find the least positive integer $b_1$ such that $x^{q+b_1}$ is in the table: $x^{q+b_1} = x^{q+a_1m}$. (Note: $0 < b_1 < m$). \\
\STATE Find the least positive integer $b_2$ such that $x^{2q+b_2}$ is in the table: $x^{2q+b_2} = x^{q+a_2m}$. (Again, $0 < b_2 < m$). \\
\STATE Compute $g= \gcd(a_1m - b_1, a_2m - b_2 - q)$. \\ 
\STATE For each divisor $d$ of $g$ below some bound $B$, do the following. \\ 
\do{ \quad \textbf{If} {$x^{N+g/d} = x^N$:} \\
\qquad set $g\leftarrow g/d$\;}
\STATE Output $L_x = g$ and stop.
\end{algorithmic}

\end{algorithm}

We first note that if $L_x>m$ and the table in Step (2) has repeated entries $x^{q+i_1m}=x^{q+i_2m}$, then numbers $b_1$ and $b_2$ may not exist below $m$. In this case the algorithm needs to be modified to take $g \leftarrow (i_1-i_2)m$. However, whenever this case does not arise, it can be shown that steps 3 and 4 are always successful in finding a collision. 

We further remark that in step 6, the list of divisors of $g$ is kept fixed, while $g$ is updated to $g/d$ whenever the condition is satisfied. In the subsequent steps, non-divisors of $g/d$ can be immediately discarded. However, the end result depends on the order in which divisors are tested, which the algorithm does not mention explicitly. However, we note that it is, in fact, possible to restrict the testing to only the prime power divisors of $g$ below $B$, and with this setting, the optimal performance is obtained by taking divisors in decreasing order. We will assume this set-up for the rest of the analysis. 

Step (2) involves $\mathcal{O}(\log N)$ steps compute $x^q$ and $x^m$ and another $\mathcal{O}(\sqrt{N})$ multiplications to compute $x^q, x^q\cdot x^m, x^q\cdot x^2m, \ldots, x^q\cdot x^{m^2}$. Step 3 involves at most $m$ multiplications $x^{q+1}=x^q\cdot x, x^{q+1}\cdot x, \ldots, x^{q+m-1}$, with complexity $\mathcal{O}(\sqrt{N})$, and match-finding with the first list, with complexity $\mathcal{O}(\sqrt{N}\log N)$ with standard sorting and search algorithms. The same is true for step (4). Step (5) has complexity $\mathcal{O}(\log \max(a_1m-b_1, a_2m-b_2-q)) = \mathcal{O}(\log N)$ and so does not contribute to the overall complexity. Step 6 involves $B$ iterations of a multiplication and an exponentiation $x^{g/d}$, and thus has a time complexity of $\mathcal{O}(B(\log g+1)) =\mathcal{O}(B\log N) $ multiplications.

In the original work, Monico states that the bound $B$ of Algorithm \ref{monico_algo} can always be chosen so that $B< \sqrt{a_1m -b_1}$. We remark that this claim does not hold in the current setting of the algorithm. For example, with a cycle length value of 4, and $a_1m-b_1 = 104$, $a_2m-b_2-q = 52$, we get $g=52$. If $B< \sqrt{a_1m-b_1} = \sqrt{104} < 11$, then we would only test divisors $d$ below 11, and would never factor out 13 to obtain the true cycle length. For such a bound to work, one needs to modify the algorithm to test both divisors $g$ and $g/d$ in step 6. However, we will show in Lemma \ref{monico_success} that it is almost always sufficient to take $B$ to be a reasonably large fixed constant, thus the complexity of step 6 can be counted as $\mathcal{O}(\log N)$, and does not contribute to the overall complexity. Thus, the overall time complexity is $\mathcal{O}(\sqrt{N}\log N)$. If $N$ is unavailable, the algorithm can also be modified to update the value of $N$ step-by-step until a large enough value is found. In this case, Algorithm~\ref{monico_algo} has a total complexity of $\mathcal{O}(\sqrt{N_x}(\log N_x)^2)$.

Further, Monico suggests a modification to the above algorithm, viz. to find several such $a_i$ and $b_i$ and compute all the gcd's. It is clear that this suggestion is exactly the method used in Banin and Tsaban's algorithm as discussed in Section \ref{banin}.

We now analyze the probability of success. The algorithm first looks for a collisions of the form $x^{q+{a_1}m} = x^{q+b_1}$. The working principle is that in this case, the cycle length $L_x$ divides $a_1m-b_1$. Similarly, if also $x^{q+{a_2}m} = x^{2q+b_2}$ then $g = \gcd(a_1m-b_1, a_2m-b_2-q)$ is a multiple of $L_x$.

 So far, the process is essentially the same in both Algorithms \ref{banin_algo} and \ref{monico_algo}: while the former uses a discrete logarithm oracle to obtain multiples of the cycle length, the latter directly finds these multiples by finding collisions. However, in Algorithm \ref{monico_algo}, we do not proceed with computing multiple factors of $L_x$, but work with the fixed multiple $g$ of $L_x$, whereas in Algorithm \ref{banin_algo} this multiple shrinks several times.

Algorithm \ref{monico_algo} then proceeds by fixing a bound $B$ and iterating over every number $d$ below $B$ to check if $d \mid g$. If yes, it executes the next part, i.e. checks if $x^{N+D/d} = x^N$, and if this holds, it sets $D \leftarrow D/d$. Note that if the factorization of the number $g$ is known (or if $g$ can be factored in time negligible compared to $O(\sqrt{N})$, then we do not need this fixed bound $B$, and can instead iterate over every prime factor $d$ of $g$. It is well-known that the number of prime factors of $g$ counted with multiplicity is $\mathcal{O}(\log g)$, so Step (5) of the algorithm can find $L_x$ in $\mathcal{O}(\log N)$ steps. However, in general, factoring $g$ may be difficult, so we assume from here on that the algorithm proceeds by fixing a bound $B$ for the divisors of $g$. Below we analyse the probability of the algorithm succeeding in terms of $B$ and $g$.

\begin{lemma} \label{monico_success}  The probability that Algorithm \ref{monico_algo} succeeds is bounded below by $\left(1-\frac{1}{B}\right)^{\log g}$.
\end{lemma}
\begin{proof}
We write $g=L_x\cdot F$ for some number $F$ and suppose that the algorithm fails. This means that there is a divisor, and hence also a prime power divisor of $F$, which the algorithm fails to factor out. Let $p$ be a prime dividing $F$, $\alpha_p$ denote its largest power dividing $F$, and $\beta_p$ be its largest power below the fixed bound $B$. So, we have $p^{\alpha_p} \mid F$, $p^{\alpha_p+1} \nmid F$, $p^{\beta_p}<B$, $p^{\beta_p+1}>B$. 

Since the number of times the algorithm divides $g$ by $p$ is $$\sum\limits_{i=1}^{\beta_p} i = \beta_p\cdot (\beta_p+1)/2,$$ we must have $\beta_p\cdot (\beta_p+1)/2 < \alpha_p$ if the algorithm fails. So, the algorithm succeeds as long as  $\beta_p\cdot (\beta_p+1)/2 \geq \alpha_p$ for every prime divisor $p$ of $F$.
Thus, the probability of success for the algorithm can be bounded below by \begin{align*}
    \prod_{p \mid g} \mathrm{Prob} \left(\frac{\beta_p\cdot(\beta_p+1)}{2} \geq \alpha_p \right).
\end{align*}

Write $ v_p= \frac{\beta_p(\beta_p+1)}{2}$ for simplicity. We may assume that $g$ is a random multiple of $L_x$ below the bound $B$, so $F$ is a random number in $\{1, \ldots, \frac{B}{L_x}\}$. We have,
\begin{align*}
 \mathrm{Prob} ( \alpha_p \leq v_p)  = &  1-\mathrm{Prob}(p^{v_p+1}\mid F) \\ =& 1-\left(\frac{B/L_x}{p^{v_p+1}(B/L_x)}\right) \\ =&  1-1/p^{v_p+1}  =   1-\dfrac{1}{p^{\frac{\beta_p(\beta_p+1)}{2}+1}}. \end{align*}
 Hence, a lower bound for the probability of the algorithm's success is $$\prod\limits_{p \mid F} \left(1-\dfrac{1}{p^{\frac{\beta_p\cdot(\beta_p+1)}{2}+1}}\right).$$

 Now, we have,  \begin{align*}
    & p^{\beta_p+1}>B \; \iff  \dfrac{1}{p^{\beta_p+1}} < \dfrac{1}{B} \\
    \implies & 1-\dfrac{1}{p^{\frac{\beta_p(\beta_p+1)}{2}+1}} > 1-\dfrac{1}{B^{\frac{\beta_p}{2}+1}} >1- \dfrac{1}{B}.
\end{align*}

We further make the following observation. Let $\omega(n)$ denote the number of distinct prime divisors of integer $n$ (note, however, that the same statement also holds if counted with multiplicity). Then clearly, $2^{\omega(n)} \leq n,$ and so, taking logarithms, $\omega(n) \leq \log_{2} n.$

Collecting all the above results, we conclude that the probability of success {Prob (success)}  of Algorithm \ref{monico_algo} is bounded below as follows.  \begin{align*}
    \text{Prob (success)} &\geq \prod_{p\mid F}\left(1- \frac{1}{B}\right) \\ &= \left(1- \frac{1}{B}\right)^{\omega(F)} \geq \left(1-\frac{1}{B}\right)^{\log F} \\ &\geq \left(1-\frac{1}{B}\right)^{\log g}.
\end{align*} \end{proof}
Note that this bound shows that Algorithm \ref{monico_algo} is indeed successful with overwhelming probability, as conjectured by the author. For example, with $B=10^6$, even when $g$ is several orders of magnitude larger than $B$, say $g = 2^{4000}$, the probability of success is greater than 99.6 percent, by the bound derived in Lemma \ref{monico_success}.

\section{Deterministic Solution of the DLP}\label{det}
 The solution of the DLP in a semigroup involves two parts: the calculation of the cycle length and start of the base element $x$, and the use of this value to find the discrete log. 
 
\subsection{Deterministic Algorithm for Cycle Length Computation}

We now present our deterministic algorithm for the computation of the cycle length. It works by finding a suitable collision, and also guarantees finding the actual cycle length rather than just a multiple of it, in a fixed number of steps. 

\begin{algorithm}[ht]
 \caption{Deterministic Algorithm for Cycle Length} \label{det_algo}

\hspace*{\algorithmicindent} \textbf{Input} {A semigroup $S$ and a torsion element $x \in S$.
Assume $N_x$ is the order of $x$.} \\
 \hspace*{\algorithmicindent} \textbf{Output} {Cycle length $L_x$ of $x$}

\begin{algorithmic}[1]

\STATE{Initialize $N \leftarrow 1$.}
\STATE{Set $q \leftarrow \lceil \sqrt{N} \rceil$. \\
\STATE{Compute, one by one, $x^N, x^{N+1}, \ldots, x^{N+q}$ and check for the equality $x^N=x^{N+j}$ at each step $j\geq 1$. Store these values in a table as pairs $(N+j, x^{N+j})$, $0 \leq j < q$. If $x^N=x^{N+j}$ for any $j<q$, then set $L_x\leftarrow j$ and end the process. If not, proceed to the next step.} \\
\STATE {For $0 \leq i \leq q$, compute, one by one, the values $x^{N+q}, x^{N+2q}, \ldots, x^{N+iq}$ and at each step $i$, look for a match in the table of values calculated in Step (3).} \\
\STATE{ Suppose that a match $x^{N+iq} = x^{N+j}$ is found, and $i$ is the smallest integer such that this happens. Set $L_x \leftarrow iq - j$ and end the process. }\\
\STATE{If no match is found in steps 3 or 5, set $N \leftarrow 2\cdot N$ and go back to Step (2).}} \\
 \end{algorithmic}
\end{algorithm}

\begin{theorem}  \label{cycle_length}
Let $S$ be a semigroup and $x\in S$ a torsion element with order $N_x$.
If an upper bound on $N_x$ is known, Algorithm~\ref{det_algo} returns the correct value of the cycle length $L_x$ with 
\[\mathcal{O}\left( \sqrt{N_x}\cdot  (\log N_x)^2\   \right)  \] steps. The total space complexity is $\mathcal{O}\left(\sqrt{N_x}\right)$ semigroup elements.
\end{theorem}
\begin{proof}

We first assume $N \geq \max(L_x, s_x)$ and show that steps 1 to 5 succeed in finding $L_x$. We have $q=\lceil\sqrt{N}\rceil$. If $L_x<q$, then the equality $x^N= x^{N+L_x}$ is found in the first step and the statement of the theorem follows. Else if $L_x \geq q$, we can write uniquely \[L_x = iq - j,\] for some positive integers $i>0$, $0 \leq j<q$. Now, we must have $i \leq q$, because otherwise if $i \geq q+1$, we would have $$L_x \geq (q+1)q - j > q^2 +q -q = q^2 \geq N,$$ a contradiction.

We have \begin{align*}
   & L_x=iq-j, \; 0<i\leq q, 0\leq j<q\\
     \implies & N+j+L_x=N+iq \\
     \implies & x^{N+j}= x^{N+j+L_x} = x^{N+iq},
\end{align*} where the last step follows because $N>s_x$ by assumption.
So, such a collision always occurs between elements of the two lists in the algorithm.

We now claim that for the smallest such integer $i$ computed in Step (5) of Algorithm \ref{det_algo}, $L_x=iq-j$. \\ 
To see this, let $i$ be the smallest positive integer such that \[x^{N+j} = x^{N+iq}.\] Also let $L_x = i^{\prime}q-j^{\prime}$, $0< i^{\prime}\leq q$, $0\leq j^{\prime}<q$. We have already shown above that such integers $i^{\prime}$ and $j^{\prime}$ exist for our choice of $N$.
By the definition of $L_x$, we must have $L_x\mid iq-j$. Now suppose that $i^{\prime}>i$. Then, \begin{align*}
    i^{\prime}q-j^{\prime} \geq & (i+1)q-j^{\prime} \\ = & iq+(q-j^{\prime}) > iq \\  \geq & iq-j.
\end{align*}
But, $L_x= i^{\prime}q-j^{\prime} \mid iq-j$, so we must have $iq-j =  i^{\prime}q-j^{\prime}$. Since $i^{\prime}>i$, this means that \[q \leq(i^{\prime}-i)q = (j^{\prime} - j) < j^{\prime},\] which is a contradiction because $0\leq j^{\prime}<q$. So, we must have $i^{\prime} = i$, $j^{\prime}=j$. This proves the claim.

We have shown above that the algorithm finds the correct cycle length when $N>\max(s_x, L_x)$. Since the algorithm doubles the value of $N$ until a match is found, it always terminates and outputs the correct cycle length. We now look at the time complexity. 

For a given $N$, step (2) involves one exponentiation, or $\mathcal{O}(\log N)$ multiplications to find $x^N$ and then at most another $q= \mathcal{O}(\sqrt{N})$ multiplications and equality checks for $x^{N}\cdot x, x^{N}\cdot x^2, \ldots, x^{N}\cdot x^q$. This step also needs a storage space of at most $q= \mathcal{O}(\sqrt{N})$ elements. Step 5 needs one exponentiation or $\mathcal{O}(\log N)$ multiplications to find $x^q$, and then another $q= \mathcal{O}(\sqrt{N})$ multiplications to find $x^{N+q}\cdot x^q, x^{N+q}\cdot x^{2q} \ldots,  x^{N+q^2}$. Finding matches in steps 3 and 5 can be done in $\mathcal{O}(q \log q)=\mathcal{O}(\sqrt{N} \log \sqrt{N})$ comparisons with the use of sorting and efficient look-up methods. Thus, clearly, steps 1 to 5 in algorithm \ref{det_algo} have a total complexity of $\mathcal{O}(\sqrt{N}\log N)$.

Moreover, the algorithm starts at $N=1$ and doubles $N$ until the cycle length is found, i.e. until $N>\max(s_x, L_x)$. Thus, the number of times steps 2 to 5 are performed is \begin{align*}
  &\left\lceil \log\left(\max\left(L_x, s_x\right)\right)\right\rceil  =  \mathcal{O}\left(\max\left( \log\left(L_x\right),  \log(s_x)\right) \right)  = \mathcal{O}(\log N_x)
\end{align*}

Thus, the total number of steps involved is 
\[\mathcal{O}\left( \sqrt{N_x} \cdot (\log N_x)^2 \right). \]
Clearly, Step (3) involves the storage of $q = \lceil \sqrt{N}\rceil = \mathcal{O}\left(\sqrt{\max(s_x, L_x)}\right) = \mathcal{O}\left(\sqrt{N_x}\right)$ elements, so this value gives the total space complexity. This completes the proof.
\end{proof}

\begin{remark} If a bound $N$ on the order $N_x$ is known a priori, then Algorithm~\ref{det_algo} can clearly be completed in a single round, with time complexity $\mathcal{O}\left( \sqrt{N} \cdot (\log N) \right)$.
\end{remark}

\begin{remark} For the case of a group, there exist better algorithms for the computation of the order of an element even when the total group order is unbounded. For instance, Algorithm 3.3 in \cite{Sutherland2007OrderCI} uses a growth function $d(t)$, which generalizes the square root function used above, to compute the order $N$ of a group element $x$, and achieves time and space complexities of $\mathcal{O}\left(\sqrt{N}\right)$, thus eliminating the additional $\log N$ multiplier introduced by the method in Algorithm \ref{det_algo}. 

However, this method fails when used for a general semigroup due to the presence of two independent unknown components of the order. To see this, note that the algorithm would need to be modified for a semigroup as follows. At stage $t$, one has $g(t-1) \leq N_x < g(t)$. On the completion of the baby steps, one has a table with the powers $x^{g(t)}, x^{g(t)+1}, \ldots, x^{g(t)+b(t)}$ (the addition of $g(t)$ is necessary in the semigroup case to ensure that the loop is entered). The giant steps compute $x^{g(t)+g(t-1) +\cdot b(t)}, x^{g(t)+g(t-1) + 2\cdot b(t)}, \ldots, x^{g(t)+g(t-1) +d(t)\cdot b(t)} =x^{2g(t)}$. Now, while $N_x$ is guaranteed to have a unique expression as $g(t-1)+ ib(t) - j$ with $0< i\leq d(t)$ and $0\leq j \leq b(t)$, this does not necessarily lead to a collision. In fact, if $b(t)<L_x <g(t-1)$ and $2L_x > g(t)=g(t-1) +d(t)\cdot b(t)$, then neither the baby steps nor the giant steps leads to a collision, and the cycle length is never found (note that this can happen only if $L_x>s_x$). Moreover, if a collision $x^{g(t)+g(t-1)+ ib(t)} = x^{g(t)+j}$ is obtained in the giant step phase, the only conclusion that can be drawn is that $L_x \mid g(t-1)+ ib(t) - j$. If instead we forced the condition $g(t-1) \leq N_x < g(t)$, a collision again may never occur because there is no control on the cycle start (For instance, in matrix semigroups over finite simple semirings, the cycle start is often found to be much larger than the cycle length. In such cases, adapting group-based algorithms would fail). See Remark 1 for further details. 
\end{remark}

\subsubsection{Experimental Results for Cycle Length Computations}
We used Algorithm \ref{det_algo} to compute cycle length values in several common semigroups, such as matrix semigroups over finite fields, matrix semigroups over the finite simple semiring $S_{20}$ (see \cite{zumbr} for a construction and \cite{ma07} for the addition and multiplication tables), and the symmetric and alternating groups (where the cycle length is precisely the order of the element).  We further used the obtained cycle lengths to compute the cycle start values using Algorithm \ref{cycle_start_algo}.  The working code may be found at  \begin{small} \href{https://github.com/simran-tinani/semigroup-cycle-length}{https://github.com/simran-tinani/semigroup-cycle-length}. \end{small}

\subsection{Solving the DLP once the Cycle Length is known}
In this section, we demonstrate the solution of the DLP for a torsion element $x$ in the semigroup $S$
once the cycle length is known.
As before let $N_x$ be the order of the sub-semigroup $\langle x \rangle$, let $L_x$ be the cycle length of the torsion element $x$ (which we assume is already computed) and let $y\in \langle x \rangle $ be an element.

In \cite{banin}, the authors demonstrate the next steps in solving for $\log_x(y)$, via a reduction to a DLP in the group $G_x$, once $L_x$ and $s_x$ are known. The procedure is described in Algorithm \ref{dlp_algo_main} below, which has been adapted from the original formulation in \cite{banin}.

\begin{algorithm}[ht]
\hspace*{\algorithmicindent} \textbf{Input} {A semigroup $S$, a torsion element $x \in S $, with cycle length $L_x$ and cycle start $s_x$, and $y\in S$ with $y=x^m$} \\
 \hspace*{\algorithmicindent} \textbf{Output} {The discrete logarithm $m$ of $y$ with base $x$}
\begin{algorithmic}[1]
\STATE Compute $t= \left\lceil\frac{s_x}{L_x}\right\rceil$ and define $x^{\prime} = x^{tL_x+1} \in G_x$. \\
\STATE Find the minimum number $0 \leq b\leq t$ such that $y^{\prime}= y\cdot x^{bL_x} \in G_x$ using binary search.  \\
\STATE Use Shank's Baby-Step Giant-Step algorithm for the group $\langle x^\prime \rangle \subseteq G_x$ to compute $m^\prime \in \{0, 1, \ldots, L_x-1\}$ such that ${(x^\prime)}^{m^\prime}= {y}^{\prime}$. \\
\STATE  Find the maximum number $c\geq 0$ such that $x^{(tL_x+1)m^{\prime}-cL_x} \in G_x$ using binary search. \\
    \STATE  Return $m = m^{\prime}(tL_x+1) - (b+c)L_x$.
\caption{Algorithm for Discrete Logarithm}\label{dlp_algo_main}
\end{algorithmic}
\end{algorithm}

Since authors of \cite{banin} do not give an explicit proof of correctness Step 5 in Algorithm \ref{dlp_algo_main}, we provide it in Theorem~\ref{algo_correctness}. Before this, we will need the following technical result.

\begin{lemma}\label{proof_gap} Let $L_x$ be the cycle length of $x \in S$, and $n$, $a$, and $a^\prime$ be fixed positive integers. Suppose that $x^{bL_x+n} = x^a \in G_x$, where $b$ is the minimum number such that $x^{bL_x+n} \in G_x$, and $x^{n-cL_x} = x^{a^{\prime}} \in G_x$, where $c$ the maximum number such that $x^{n-cL_x} \in G_x$. Then \[bL_x+n \leq a, \ \text{and} \ n-cL_x \leq a^{\prime}.\]
\end{lemma}

\begin{proof}
First let $x^{bL_x+n} = x^a$ with $b$ minimal such that $x^{bL_x+n}\in G_x$. Suppose, to the contrary, that $bL_x+n>a$. We must have, by the minimality of $b$, $x^{(b-1)L_x + n} \not\in G_x$, so $(b-1)L_x + n <a$. \begin{align*}
\text{But}, \; \; \;    & x^{bL_x+n} = x^a \in G_x \\
    \implies & bL_x+n-a = kL_x, \ k \geq 1 \\
    \implies & (b-k)L_x + n = a \\
    \implies & x^{(b-k)L_x + n} = x^a \in G_x, \ k \geq 1.
\end{align*} This is a contradiction to the minimality of $b$. So, $bL_x+n \leq a$.
Now suppose that $ x^{x-cL_x} = x^a \in G_x$, with $c$ maximal, and suppose that $n-cL_x >a^\prime$. We argue as above:
\begin{align*}
   & L_x \mid n-cL_x-a^\prime \\
\implies & n-(k+c)L_x=a^\prime, \ \text{for some} \ k \geq 1 \\
  \implies & x^{n-(k+c)L_x} = x^{a^\prime} \in G_x,
\end{align*} which is a contradiction to the maximality of $c$. Thus $n-cL_x \leq a^{\prime}$.
\end{proof}

\begin{theorem}\label{algo_correctness}
Let $S$ be a semigroup, $x\in S$ a torsion element and 
 $y\in \langle x \rangle $  any element. Assume the cycle length $L_x$ and cycle start $s_x$ of $x$
 are known. Then 
 Algorithm \ref{dlp_algo_main} returns the correct values of the discrete logarithm $m=\log_x(y)$ in $\mathcal{O}\left(  \sqrt{L_x}+(\log N_x)^2 \right)$ semigroup multiplications,
 with a required storage of $\mathcal{O}\left( \sqrt{L_x} \right) $
semigroup elements.
\end{theorem}
\begin{proof}
We use the notations of Algorithm~\ref{dlp_algo_main}, and also write $n=\log_x y$. We will show that the output $m$ is equal to the correct discrete logarithm value $n$.
 Recall that we have a group $G_x$, generated by ${x}^\prime:=x^{tL_x+1}$, and with identity $x^{tL_x}$. The parameter $t$ is given by the formula $t=\left\lceil\frac{s_x}{L_x}\right\rceil$. Inverses in $G_x$ can be computed in polynomial time using the formula \eqref{inverseformula}. Note that membership in $G_x$ can be tested with one equality check: $y\in G_x \iff y\cdot x^{L_x} =y$.
 There are now two cases:
 \begin{enumerate}
     \item 
When $y\in G_x$, we have $b=0$. Here, it is possible to use Shank's Baby Step-Giant Step algorithm \cite{shanks1971class} which is a deterministic algorithm and which requires  $\mathcal{O} \left(\sqrt{L_x} \right)$ semigroup multiplications and storage space $\mathcal{O} \left(\sqrt{L_x} \right)$, in order to compute $\log_{{x}^\prime}(y)$. This is done in Step (3). From this value, $n=\log_x(y)$ is readily computed, as shown below. Note that in this case, $\log_x(y)$ is determined modulo $L_x$. 
   \item When $y\not\in G_x$, Algorithm~\ref{dlp_algo_main} first computes, using binary search, the smallest power $b$ of $x^{L_x}$ such that the product $y\cdot x^{bL_x}$ lies in the group $G_x$, and then proceeds as in case~1 via the Baby Step-Giant Step algorithm to find the discrete logarithm $m^{\prime}$ of $y\cdot x^{bL_x}$ with base $x^\prime$ (i.e. ${(x^\prime)}^{m^\prime} = y\cdot x^{bL_x}$). Note that in this case, the value of $\log_x(y)$ is less than $s_x$, and is thus determined uniquely in $\N$. Again, the time and space complexity are both $\mathcal{O} \left(\sqrt{L_x} \right)$.
\end{enumerate}

In both cases above, we have the maximal value $c$ such that $x^{m^{\prime}(tL_x+1)-cL_x} \in G_x$, and so $c \leq L_x + s_x+1 = N_x+1$, since $m^\prime \leq L_x$ and $tL_x \leq L_x+ s_x$. We also clearly have $b \leq t \leq N_x$. Since the computations of both $b$ and $c$ are done via binary searches, they contribute $\mathcal{O}((\log N_x)^2)$ steps to the overall time complexity. Now, \[x^{m^{\prime}(tL_x+1)-cL_x} = x^{m^{\prime}(tL_x+1)} = (x^\prime)^{m^\prime}  = x^{bL_x+ n}.\]
Applying Lemma \ref{proof_gap} to the above equation, we must have \[m^{\prime}(tL_x+1)-cL_x \leq bL_x+n, \ \text{and} \ bL_x+n \leq m^{\prime}(tL_x+1)-cL_x.\]
Therefore, \ $bL_x+n = m^{\prime}(tL_x+1)-cL_x$, or \ $n = m^{\prime}(tL_x+1) - (b+c)L_x,$ which is precisely equal to $m$, the value returned by the Algorithm~\ref{dlp_algo_main}. Thus, $m=n$. This completes the proof.
\end{proof}

Combining Theorem~\ref{cycle_length}, Lemma~\ref{cycle_start}  and Theorem~\ref{algo_correctness}
we arrive at the main proposition of the paper:

\begin{proposition}\label{main_thm}
Let $S$ be a semigroup, $x\in S$ a torsion element and 
 $y\in \langle x \rangle $  any element.
The discrete logarithm $m=\log_x(y)$ can be computed deterministically in 
\[\mathcal{O}\left( \sqrt{N_x}\cdot  (\log N_x)^2\   \right)  \]
steps, with a required storage of $\mathcal{O}\left( \sqrt{N_x} \right) $
semigroup elements.
\end{proposition}

\begin{proof}
For the solution, one begins by finding $L_x$. This can be done using Algorithm \ref{det_algo} 
and according to Theorem~\ref{cycle_length} this requires
$\mathcal{O}\left( \sqrt{N_x} \cdot (\log N_x)^2 \right)$ steps and the storage of $\mathcal{O}\left( \sqrt{N_x} \right)$ elements.

By Lemma~\ref{cycle_start} the computation of the cycle start $s_x$ is achieved in 
$\mathcal{O}((\log N_x)^2)$ semigroup multiplications, which does not contribute to the overall cost 
of the algorithm.

By Theorem~\ref{algo_correctness}, the discrete logarithm $m$ can then be retrieved using Algorithm \ref{dlp_algo_main}, in $\mathcal{O}\left((\log N_x)^2 + \sqrt{L_x} \right)$ steps, with a required storage of $\mathcal{O}\left( \sqrt{L_x} \right)$ semigroup elements.

As $L_x\leq N_x$, the overall complexity is dominated by the computation of the cycle length, and the proof of the result is now clear.
\end{proof}

\subsection{Solving the DLP once the Factorization of the Cycle Length is known}\label{pohlig-sub}

We mentioned in the introduction that for a general group of order $N$ the best general known
algorithms  for solving the discrete logarithm problem have complexity  $\mathcal{O}(\sqrt{N})$ operations. 

In case the order  $N$  has a prime factorization into small primes there is the 
famous  Pohlig--Hellman algorithm \cite{pohlig1978improved} for solving the 
DLP whose complexity is dominated by the largest prime factor in the integer factorization 
of  $N$.

In case that we have available the integer factorization of the cycle length $L_x$
we can adapt the Pohlig--Hellman algorithm for groups to a Pohlig--Hellman algorithm
for solving the DLP in a semigroup.  Algorithm~\ref{dlp_algo2} represents this adapted 
Pohlig--Hellman algorithm.

\begin{algorithm}[H]
\hspace*{\algorithmicindent} \textbf{Input} {A semigroup $S$, a torsion element $x \in S $, with cycle length $L_x = \prod_{i=1}^r p_i^{e_i}$ and cycle start $s_x$, and $y\in S$ with $y=x^m$} \\
 \hspace*{\algorithmicindent} \textbf{Output} {The discrete logarithm $m$ of $y$ with base $x$}
\begin{algorithmic}[1]
\STATE Compute $t= \left\lceil\frac{s_x}{L_x}\right\rceil$ and define $x^{\prime} = x^{tL_x+1} \in G_x$. \\
\STATE Find the minimum number $0 \leq b\leq t$ such that $y^{\prime}= y\cdot x^{bL_x} \in G_x$ using binary search.  \\
\STATE \textbf{for $i\in\{1, \ldots, r\}$}
{\begin{enumerate}
\item Compute the values $x_i^{\prime} = (x^{\prime})^{L_x/p_i^{e_i}}$, \ $y_i^{\prime} = (y^{\prime})^{L_x/p_i^{e_i}}$, and \ $\gamma_i := (x_i^{\prime})^{p^{e_i-1}}$. 
\item Calculate the inverse $z_i$ of ${x_i}^\prime$ in $G_x$ using \eqref{inverseformula}. \\
\item Set $k\leftarrow 0$ and $n_0\leftarrow 0 $. \\
\item \textbf{while $k < e_i$ do} {\begin{enumerate}
     \item Compute \ ${y_k^{\prime}} = (y_i^{\prime}z_i^{n_k})^{p^{e_i-1-k}} \in \langle \gamma_i \rangle$. \\
    \item Use Shank's Baby-Step Giant-Step algorithm for the group $\langle \gamma_i \rangle \subseteq G_x$ to compute $d_k \in \{0, 1, \ldots, p_i-1\}$ \\
    such that ${\gamma_i}^{d_k} = {y_k}^{\prime}$. \\
    \item Set $n_{k+1}\leftarrow n_k+p_i^kd_k$, and $k \leftarrow k+1$. \\
     \end{enumerate}} 
     \item \textbf{end while}
\item Set $m_i:=n_{e_i}$. \\
\end{enumerate}}
\STATE \textbf{end for}
\STATE Use the Chinese Remainder Theorem to solve the congruence equations \[{\displaystyle m^{\prime}\equiv m_{i}{\pmod {p_{i}^{e_{i}}},}\quad \forall \ i \in \{1,\dots ,r\}}\] uniquely for $m^{\prime} \mod L_x$. This gives the discrete logarithm of $y^{\prime}$ with respect to the base $x^{\prime}$ in the group $G_x$. \\
\STATE  Find the maximum number $c\geq 0$ such that $x^{(tL_x+1)m^{\prime}-cL_x} \in G_x$ using binary search. \\
    \STATE  Return $m = m^{\prime}(tL_x+1) - (b+c)L_x$.
\caption{Pohlig--Hellman Algorithm for solving the Discrete Logarithm Problem in a Semigroup}\label{dlp_algo2}
\end{algorithmic}
\end{algorithm}

\begin{theorem}\label{algo_final}
Let $S$ be a semigroup, $x\in S$ a torsion element and 
 $y\in \langle x \rangle $  any element. Assume the cycle start $s_x$ of $x$ is known 
 and assume the integer factorization of the cycle length $L_x$ is known to be $L_x = \prod_{i=1}^r p_i^{e_i}$. Then  Algorithm~\ref{dlp_algo2} computes the discrete logarithm $\log_x y$ requiring
$\mathcal{O}\left(  \sum\limits_{i=1}^{r} e_i \left(\log L_x + \sqrt{p_i}\right) +\left(\log N_x \right)^2 \right)$ steps. The space complexity of the algorithm
consists in $\mathcal{O}\left(\sum\limits_{i=1}^{r} e_i \sqrt{p_i}\right)$ semigroup elements.
\end{theorem}

\begin{proof}
Steps 1 and 2 are in analogy to the corresponding steps of Algorithm~\ref{dlp_algo_main}.
Steps 3 to 5  represent the Pohlig--Hellman algorithm for groups with the implied 
complexity dominated by the largest prime factor $p_i$ of the integer factorization of $L_x$ (for a reference on Pohlig--Hellman in groups, see  in \cite[Theorem 2.32]{hoffstein2008introduction}).
It follows that the running time of the algorithm is $\mathcal{O}\left(  \sum\limits_{i=1}^{r} e_i \left(\log L_x + \sqrt{p_i}\right) \right)$ steps. The computation of $b$ and $c$
require in addition $\left(\log N_x \right)^2 $ steps.
The total space complexity is $\mathcal{O}\left(\sum\limits_{i=1}^{r} e_i \sqrt{p_i}\right)$ semigroup elements and that completes the proof.
\end{proof}

\section{Conclusion} The DLP in a finite group has noteworthy significance for cryptography, and so an extension of existing solutions to other algebraic structures like semigroups, where inverses may not be available, is of natural interest. In particular, the DLP in a semigroup has been discussed before in two places, namely \cite{banin} and \cite{monico2002}. Both these authors provide probabilistic generalizations of existing collision-based methods for the case of a semigroup. The time complexity of the algorithm in \cite{banin} is $\mathcal{O}(\sqrt{N_x}(\log N_x)^2\log\log N_x)$, and the one in \cite{monico2002} is $\mathcal{O}(\sqrt{N_x}(\log N_x)^2)$.  Both these methods rely on computing a multiple of the cycle length and then taking gcd's or factors, and could fail with a small probability that depends on the parameters chosen. In this paper, we provided a deterministic solution of the semigroup DLP, which computes the cycle length directly and does not rely on finding a factor of it. The time complexity of our algorithm is $\mathcal{O}(\sqrt{N_x}(\log N_x)^2)$. We further demonstrated the application of the Pohlig-Hellman algorithm to semigroups. A direct consequence of our findings is that for cryptographic purposes, generalizing the type of algebraic structure for the DLP offers no additional advantage, at least in the torsion case, both in a classical and a quantum setting.

\section{Acknowledgement}

This research is supported by armasuisse Science and Technology. The second author is also supported by Swiss National Science Foundation grant n. 188430.

\bibliographystyle{plain}
\bibliography{huge1}

\def\cprime{$'$} \def\polhk#1{\setbox0=\hbox{#1}{\ooalign{\hidewidth
  \lower1.5ex\hbox{`}\hidewidth\crcr\unhbox0}}}
  \def\polhk#1{\setbox0=\hbox{#1}{\ooalign{\hidewidth
  \lower1.5ex\hbox{`}\hidewidth\crcr\unhbox0}}} \def\cprime{$'$}
  \def\cprime{$'$} \def\cprime{$'$} \def\cprime{$'$}
\begin{thebibliography}{10}

\bibitem{ad93}
L.~M. Adleman and J.~DeMarrais.
\newblock A subexponential algorithm for discrete logarithms over all finite
  fields.
\newblock {\em Math. Comp.}, 61(203):1--15, 1993.

\bibitem{ag04}
M.~Agrawal, N.~Kayal, and N.~Saxena.
\newblock P{RIMES} is in {P}.
\newblock {\em Ann. of Math. (2)}, 160(2):781--793, 2004.

\bibitem{banin}
M.~Banin and B.~Tsaban.
\newblock A reduction of semigroup {DLP} to classic {DLP}.
\newblock {\em Des. Codes Cryptography}, 81(1):75–82, October 2016.

\bibitem{quantum}
A.M. Childs and G.~Ivanyos.
\newblock Quantum computation of discrete logarithms in semigroups.
\newblock {\em Journal of Mathematical Cryptology}, 8(4):405 -- 416, 01 Dec.
  2014.

\bibitem{co06}
H.~Cohen, G.~Frey, R.~Avanzi, C.~Doche, T.~Lange, K.~Nguyen, and
  F.~Vercauteren, editors.
\newblock {\em Handbook of Elliptic and Hyperelliptic Curve Cryptography}.
\newblock Discrete Mathematics and its Applications (Boca Raton). Chapman \&
  Hall/CRC, Boca Raton, FL, 2006.

\bibitem{cormen2009introduction}
Thomas~H Cormen, Charles~E Leiserson, Ronald~L Rivest, and Clifford Stein.
\newblock {\em Introduction to algorithms}.
\newblock MIT press, 2009.

\bibitem{di76}
W.~Diffie and M.~E. Hellman.
\newblock New directions in cryptography.
\newblock {\em IEEE Trans. Inform. Theory}, IT-22(6):644--654, 1976.

\bibitem{go20}
N.~Goel, I.~Gupta, and B.~K. Dass.
\newblock Survey on {SAP} and its application in public-key cryptography.
\newblock {\em J. Math. Cryptol.}, 14(1):144--152, 2020.

\bibitem{hoffstein2008introduction}
Jeffrey Hoffstein, Jill Pipher, Joseph~H Silverman, and Joseph~H Silverman.
\newblock {\em An introduction to mathematical cryptography}, volume~1.
\newblock Springer, 2008.

\bibitem{ka13}
D.~Kahrobaei, C.~Koupparis, and V.~Shpilrain.
\newblock Public key exchange using matrices over group rings.
\newblock {\em Groups Complex. Cryptol.}, 5(1):97--115, 2013.

\bibitem{ma07}
G.~Maze, C.~Monico, and J.~Rosenthal.
\newblock Public key cryptography based on semigroup actions.
\newblock {\em Adv. in Math. of Communications}, 1(4):489--507, 2007.

\bibitem{me97}
A.~J. Menezes, P.~C. van Oorschot, and S.~A. Vanstone.
\newblock {\em Handbook of applied cryptography}.
\newblock CRC Press Series on Discrete Mathematics and its Applications. CRC
  Press, Boca Raton, FL, 1997.

\bibitem{mi76}
G.~L. Miller.
\newblock Riemann's hypothesis and tests for primality.
\newblock {\em J. Comput. System Sci.}, 13(3):300--317, 1976.

\bibitem{monico2002}
C.~Monico.
\newblock {\em Semirings and semigroup actions in public-key cryptography}.
\newblock PhD thesis, University of Notre Dame Notre Dame, 2002.

\bibitem{pohlig1978improved}
S.~Pohlig and M.~Hellman.
\newblock An improved algorithm for computing logarithms over {GF(p)} and its
  cryptographic significance (corresp.).
\newblock {\em IEEE Transactions on information Theory}, 24(1):106--110, 1978.

\bibitem{pollard1978monte}
J.~M. Pollard.
\newblock Monte {C}arlo methods for index computation.
\newblock {\em Mathematics of computation}, 32(143):918--924, 1978.

\bibitem{ra80}
M.~O. Rabin.
\newblock Probabilistic algorithm for testing primality.
\newblock {\em J. Number Theory}, 12(1):128--138, 1980.

\bibitem{shanks1971class}
D.~Shanks.
\newblock Class number, a theory of factorization, and genera.
\newblock In {\em Proc. of Symp. Math. Soc., 1971}, volume~20, pages 41--440,
  1971.

\bibitem{Sutherland2007OrderCI}
Andrew~V. Sutherland.
\newblock Order computations in generic groups.
\newblock 2007.

\bibitem{zumbr}
J.~Zumbr{\"a}gel.
\newblock Classification of finite congruence-simple semirings with zero.
\newblock {\em Journal of Algebra and Its Applications}, 07(03):363--377, 2008.

\end{thebibliography}

\end{document}